\documentclass[letterpaper, 10 pt, conference]{ieeeconf}  

\IEEEoverridecommandlockouts                              
\usepackage{amsmath,amssymb,physics,mathtools} 
\usepackage{xcolor}
\usepackage{pgfplots}
\usepackage{pgfplotstable}
\usepgfplotslibrary{fillbetween}
\usepackage[capitalize]{cleveref}
\let\labelindent\relax
\usepackage{enumitem}

\usepackage{amsthm}
\usepackage{algorithmicx}
\usepackage{algorithm}
\usepackage{algpseudocode}
\usepackage{placeins} 

\usepackage{url}
\usepackage{mdframed}
\usepackage{cite}
\makeatletter
\def\@cite#1#2{[{#1\if@tempswa , #2\fi}]}

\makeatother
\usepackage{subcaption}
\usepackage{wrapfig}
\usepackage[tracking=false,kerning=true,spacing=true]{microtype}

\usepackage{siunitx}
\DeclareSIUnit{\watthour}{Wh}

\usepackage{tikz,pgfplots}
\usetikzlibrary{arrows.meta, positioning, calc}
\pgfplotsset{every axis/.append style={
    tick label style={font=\footnotesize},
    legend style={font=\footnotesize},
    label style={font=\small}
}}

\usepackage{pgfplotstable} 
\usepgfplotslibrary{groupplots}
\pgfplotsset{compat=1.18}
\definecolor{ETHblue}{RGB}{0,112,192}      
\definecolor{ETHpurple}{RGB}{128,0,128}      
\definecolor{ETHpetrol}{RGB}{0,153,153}      
\definecolor{ETHgreen}{RGB}{112,173,71}      
\definecolor{ETHbronze}{RGB}{112,79,18}      
\definecolor{ETHrot}{RGB}{150,39,45}         
\definecolor{ETHpurpur}{RGB}{140,10,89}       
\definecolor{ETHgrau}{RGB}{87,87,87}          

\newtheoremstyle{DefinitionStyle}
{5pt}{5pt}                
{}                        
{}                        
{\bfseries}               
{}                       
{3pt plus 1pt minus 1pt}  
{%
  \thmname{#1}\thmnumber{ #2 }\thmnote{\normalfont(#3)}
  \;
}
\theoremstyle{DefinitionStyle}
\newtheorem{definition}{Definition}

\newtheoremstyle{AssumptionStyle}
  {5pt}{5pt}                
  {}                        
  {}                        
  {\bfseries}               
  {}                        
  {3pt plus 1pt minus 1pt}  
  {%
    \thmname{#1}\thmnumber{ #2 }\thmnote{\normalfont(#3)}
    \;
  }

\theoremstyle{AssumptionStyle}
\newtheorem{assumption}{Assumption}
\crefname{assumption}{Assumption}{Assumptions}

\newtheoremstyle{Theorem}
  {5pt}{5pt}                
  {}                        
  {}                        
  {\bfseries}               
  {}                        
  {3pt plus 1pt minus 1pt}  
  {%
    \thmname{#1}\thmnumber{ #2 }\thmnote{\normalfont(#3)}
    \;
  }

\theoremstyle{TheoremStyle}
\newtheorem{theorem}{Theorem}
\newtheoremstyle{PropositionStyle}
  {5pt}{5pt}                
  {}                        
  {}                        
  {\bfseries}               
  {}                        
  {3pt plus 1pt minus 1pt}  
  {%
    \thmname{#1}\thmnumber{ #2 }\thmnote{\normalfont(#3)}
    \;
  }

\newtheorem{proposition}{Proposition}
\newtheoremstyle{CorollaryStyle}
  {5pt}{5pt}                
  {}                        
  {}                        
  {\bfseries}               
  {}                        
  {3pt plus 1pt minus 1pt}  
  {%
    \thmname{#1}\thmnumber{ #2 }\thmnote{\normalfont(#3)}
    \;
  }

\newtheorem{corollary}{Corollary}
\newtheoremstyle{Examplestyle}
  {5pt}{5pt}                
  {\normalfont}             
  {}                        
  {\bfseries}               
  {}                        
  {3pt plus 1pt minus 1pt}  
  {%
    \thmname{#1}\thmnumber{ #2 }\thmnote{\normalfont(#3)}
    \;
  }

\theoremstyle{Examplestyle}

\DeclareMathOperator*{\argmin}{arg\,min}
\DeclareMathOperator*{\argmax}{arg\,max}

\newcommand{\ThetaSet}{\Theta}
\newcommand{\X}{\mathcal{X}}
\newcommand{\Xj}{\mathcal{X}^j}
\newcommand{\Xij}{\mathcal{X}_i^j}

\newcommand{\tht}{\theta}

\newcommand{\x}{x}

\newcommand{\xhatj}{\hat{x}^j}
\newcommand{\xstar}{x^\star}

\newcommand{\s}{s}
\newcommand{\sj}{s^j}

\newcommand{\Ui}{U_i}
\newcommand{\phii}{\varphi_i}
\newcommand{\F}{F}

\newcommand{\obsj}{(\xhatj,\sj,\Xj)}
\newcommand{\data}{\{(\hat{x}^j,s^j,\mathcal{X}^j)\}_{j=1}^N}

\newcommand{\lpred}{\ell_{\mathrm{pred}}}
\newcommand{\lvi}{\ell_{\mathrm{VI}}}
\newcommand{\lsub}{\ell_{\mathrm{sub}}}

\newcommand{\lsubfull}{\lsub(\tht;\hat{x}^{1:N},s^{1:N})}

\title{\LARGE \bf
Suboptimality Loss for Inverse Learning from Imperfect Equilibria
}

\author{
Andreas Feik$^{1}$, Pierre Pinson$^{2}$, and Dario Paccagnan$^{3}$%
\thanks{This paper was partially supported by the EPSRC grant EP/Y001001/1, funded by the International Science Partnerships Fund (ISPF) and UKRI; and by an Imperial-MIT seed fund.}%
\thanks{$^{1}$A. Feik is with the Automatic Control Laboratory, ETH Zurich, 8092 Zurich, Switzerland
{\tt\small andreas.feik@outlook.com}}%
\thanks{$^{2}$P. Pinson is with the Dyson School of Engineering and Design, Imperial College London, London, UK
{\tt\small p.pinson@imperial.ac.uk}}%
\thanks{$^{3}$D. Paccagnan is with the Department of Computing, Imperial College London, London, UK
{\tt\small d.paccagnan@imperial.ac.uk}}%
}
\begin{document}

\maketitle
\thispagestyle{empty}
\pagestyle{empty}

\begin{abstract}
Many modern systems involve the strategic interaction of multiple agents. In such settings, observed actions typically reflect equilibrium behavior under utilities that are only partially known. Recovering these hidden utilities from data -- the central goal of inverse game theory -- is key for prediction, counterfactual analysis, and mechanism design. However, existing approaches based on inverse variational inequalities are highly sensitive to noisy and inconsistent equilibrium observations, thus limiting their applicability. 
In this paper, we resolve this issue by introducing a game-theoretic suboptimality loss that measures the aggregate utility gain players could obtain by unilaterally deviating from an observed strategy profile. First, we show that this loss is convex and admits an efficient decomposition into player-wise best-responses. Second, we show this loss is sandwiched between the predictability loss and the inverse variational inequality loss, making it a tractable surrogate for equilibrium prediction. Third, we develop a mirror descent algorithm to minimize it and demonstrate on a heterogeneous networked Cournot competition that our approach remains accurate under noisy observations and inconsistent equilibrium data while inverse variational \mbox{inequality methods produce degenerate estimates.}

\end{abstract}


\section{Introduction}\label{sec:introduction}

Many modern socio-technical systems are shaped by the interaction of multiple strategic agents. Examples include traffic routing, electricity markets, online platforms, and resource allocation in networks \cite{paccagnan2018nash,bramoulle2014strategic}. In such settings, observed actions typically reflect (near) equilibrium behavior under utility or cost models that are only partially known. Recovering these hidden primitives from data is the central goal of \emph{inverse game theory} \cite{kuleshov2015inverse,bertsimas2015data,thai2018imputing, cui2025inverse}.  This field has received increased attention over the last years, 
because of the wealth of potential applications (transportation, 
energy, etc.) and potential impact (e.g., better coordination, 
market monitoring), as well as for its usability in prediction tasks, counterfactual analysis, and 
mechanism design. 

The most common approach to \emph{inverse game theory} is based 
on inverse variational inequalities, where unknown utility 
parameters are estimated by enforcing first-order equilibrium conditions 
\cite{bertsimas2015data}. While attractive for its 
generality and computational tractability, this approach is known 
to be highly sensitive to noisy or inconsistent equilibrium data. 
In practice, observations are invariably subject to measurement 
error, limited rationality, or inexact implementation, and even 
small perturbations can lead to degenerate parameter estimates. The consequences are significant, as 
downstream tasks such as forecasting agent behavior, designing 
incentives, or informing decisions all hinge on the 
quality of the recovered model.

In the single-agent setting, the predictability loss of 
\cite{aswani2018inverse} -- which quantifies, for a given set 
of parameters, the difference between the predicted and observed 
solution -- addresses this issue by optimizing directly for 
out-of-sample prediction performance. However, even in the 
single-agent case this leads to a bilevel optimization problem 
that is computationally demanding. The difficulty compounds 
significantly in the multi-agent setting, where this formulation leads to a mathematical program with equilibrium 
constraints, rendering it largely intractable in practice.
\medskip

In this paper, we pursue a different route: rather than 
enforcing first-order conditions or optimizing for prediction directly, we build a loss function grounded in the Nash equilibrium definition itself, that is both \emph{computationally efficient} and \emph{robust to noisy observations}. 
The key idea is simple: if an observed strategy profile is an equilibrium under a candidate model, then no player should be able to improve its utility through a unilateral deviation. 
Motivated by this principle, we introduce a \emph{game-theoretic suboptimality loss} that measures the aggregate gain players could obtain by best responding to the observed actions of the others. 
While related to~\cite{mohajerin2018data, maddux2023data}, our approach addresses multi-agent settings, with possibly shared parameters and continuous action spaces, and is the first to do so tractably and with robustness to noisy equilibrium observations. Specifically, our contributions are threefold.
\medskip



\begin{mdframed}[hidealllines=true,backgroundcolor=blue!10]
\textbf{C1.} We introduce a suboptimality loss 
for inverse game theory, defined as the aggregate gain players could obtain by unilaterally deviating from an 
observed strategy, averaged over all available samples. 
For the commonly encountered utility models that are linear in the unknown parameter, we show that the proposed loss is convex.
\end{mdframed}

\begin{mdframed}[hidealllines=true,backgroundcolor=blue!10]
\textbf{C2.} We show that, under standard assumptions, the proposed suboptimality loss is sandwiched 
between the predictability loss and the inverse variational 
inequality loss, thus serving as a tractable surrogate for the 
former. We further specialize this result to games where the utilities are linear or quadratic in the players' decisions,
yielding a more precise sandwiching \mbox{than in the general case.}%
\end{mdframed}

\begin{mdframed}[hidealllines=true,backgroundcolor=blue!10]
\textbf{C3.} We develop a mirror descent algorithm to minimize the proposed loss, relying solely on player-wise best-response. We numerically demonstrate on a networked Cournot competition that the proposed loss is substantially more robust to noisy observations than the state of the art methods, both theoretically and empirically.
\end{mdframed}

\section{Game-theoretic Suboptimality Loss}\label{sec:sub}

\subsection{Inverse Game Setup}
We consider a game with $P$ players, where an exogenous signal 
$\s$ determines the environment. For a given signal $\s$, each 
player $i$ selects an action $x_i$ from a 
feasible set $\X_i(\s)$, and receives a utility 
$\Ui(\x,\s;\tht)$ that depends on the joint action profile 
$\x\in\mathcal{X}(\s)=\mathcal{X}_1(\s)\times\cdots\times
\mathcal{X}_P(\s)$, on the signal $\s$, and on an unknown 
parameter vector $\tht\in\ThetaSet$.

We are given a dataset 
$\data$ containing, for each observation $j$, 
a signal $\sj$, a feasible set 
$\Xj=\X(s^j)$, and an observed strategy 
profile $\xhatj\in\Xj$. The goal is to reconstruct $\tht\in\ThetaSet$ such 
that all observed profiles in the dataset are well explained as approximate equilibria of the corresponding game. This problem is known as \emph{contextual inverse game learning}.

As common in the literature \cite{besbes2025contextual,mohajerin2018data}, we focus on settings where the utility is linear in the unknown 
parameter, that is,
\[
\Ui(\x,\s;\tht)=\tht^\top \phii(\x,\s),
\qquad i=1,\dots,P,
\]
with known feature maps $\phii$. We stress that linearity in 
$\tht$ does not imply linearity in the decisions; for instance, games with quadratic utilities fall well within this class. We study this setting under the following standard convexity assumptions.
\begin{assumption}[Convex game]\label{ass:convex_game}
For each player $i\in\{1,\ldots,P\}$:
\begin{enumerate}
    \item For every $\s\in\mathcal S$, the feasible set $\X_i(\s)$ is non-empty, compact, and convex.
    \item For every $\s\in\mathcal S$, every $x_{-i}\in \X_{-i}(\s)$, and every $\tht\in\ThetaSet$, the mapping
    \[
    x_i \mapsto \Ui(x_i,x_{-i},\s;\tht)
    \]
    is concave and continuous on $\X_i(\s)$.
    \item The utility $\Ui(\x,\s;\tht)$ is continuous in $(\x,\s,\tht)$.
\end{enumerate}
\end{assumption}
Beyond convex games, the loss remains meaningful whenever best responses exist, but tractability and the guarantees below may be lost.
\subsection{Measuring Suboptimality}

We now turn our attention to measuring equilibrium suboptimality. 
Recall that $\xhatj$ is a Nash equilibrium under $\tht$ if and 
only if no player can improve its utility by deviating 
unilaterally, that is,
\[
\Ui(\xhatj,\sj;\tht)\ge \Ui(x_i,\xhatj_{-i},\sj;\tht),
\quad \forall x_i\in\Xij,\ \forall i.
\]
This motivates measuring how well an observed profile $\xhatj$ 
is explained as an equilibrium under $\tht$ through the 
aggregate utility gain players could obtain by deviating 
unilaterally from~it.

\begin{definition}[Game-theoretic suboptimality loss]\label{def:sub_loss}
For an observation $\obsj$, define
\[
\lsub(\tht;\xhatj,\sj)
=
\sum_{i=1}^P
\max_{x_i\in\Xij}
\Big[
\Ui(x_i,\xhatj_{-i},\sj;\tht)-\Ui(\xhatj,\sj;\tht)
\Big].
\]
Equivalently, under the linear parametrization,
\[
\lsub(\tht;\xhatj,\sj)
=
\sum_{i=1}^P
\max_{x_i\in\Xij}
\tht^\top\!\left[
\phii(x_i,\xhatj_{-i},\sj)-\phii(\xhatj,\sj)
\right].
\]
The empirical loss over the entire dataset is therefore
\[
\lsubfull
=
\frac{1}{N}\sum_{j=1}^N \lsub(\tht;\xhatj,\sj).
\]
\end{definition}


Three important observations are in order. 
First, the loss $\lsub(\tht;\xhatj,\sj)$ vanishes if and only if no player  can profit from a unilateral deviation at $\xhatj$, making it a natural measure of equilibrium consistency.
Second, because the utilities are linear in $\tht$, the loss is positively homogeneous: $\lsub(c\tht;\xhatj,\sj)=c\,\lsub(\tht;\xhatj,\sj)$ 
for all $c\ge 0$, so that $\tht$ is identifiable only up to positive scaling. In practice, this is resolved by imposing a convex normalization constraint, such as $\tht \ge 0$ and $\mathbf{1}^\top \tht = 1$.
\footnote{We emphasize that normalization removes only scale ambiguity. Exact identification from noiseless data additionally requires the parameter to be uniquely determined by the observed context equilibrium pairs, which we do not assume. Thus, perfect prediction can coexist with nonzero parameter recovery error. This is consistent with the broader inverse optimization literature, where prediction performance is the natural measure of success.}
Third, the suboptimality loss is convex as we now show, making it amenable to efficient optimization -- see Section \ref{sec:algorithm}. 

\begin{proposition}[Convexity]\label{prop:sub_convex}
Assume that utilities are linear in $\tht$. Then, for 
every observation $\obsj$, the mapping
\[
\tht\mapsto \lsub(\tht;\xhatj,\sj)
\]
is convex on $\ThetaSet$. Consequently, the empirical loss 
$\lsubfull$ is convex in the parameter $\theta$.
\end{proposition}
\begin{proof}
For each player $i$ and feasible deviation $x_i\in\Xij$, the 
quantity
\[
\tht^\top\!\left[
\phii(x_i,\xhatj_{-i},\sj)-\phii(\xhatj,\sj)
\right]
\]
is affine in $\tht$. Taking a pointwise maximum over 
$x_i\in\Xij$ and summing over $i$ and 
$j$ preserves convexity.
\end{proof}

\subsection{Relation to Existing Losses}

The goal of this section is to show that the proposed 
suboptimality loss is a principled surrogate for 
equilibrium prediction: minimizing $\lsub$ over $\tht$ 
drives the prediction error to zero. We establish this 
by sandwiching $\lsub$ between the predictability loss 
and the inverse variational inequality loss, two 
reference losses we now recall.

The inverse variational inequality loss of  \cite{bertsimas2015data} 
is defined as
\[
\lvi(\tht;\xhatj,\sj)
:=
\max_{\x\in\Xj} \F(\tht;\xhatj,\sj)^\top(\xhatj-\x),
\]
where $\F(\tht;\xhatj,\sj)$ denotes the pseudo-gradient\footnote{The pseudo-gradient is the stacked vector of players' negative partial gradients with respect to their own actions, i.e.,
\(
\F(\tht;\x,\s)=
\bigl(
-\nabla_{x_1}U_1(\x,\s;\tht),\dots,-\nabla_{x_P}U_P(\x,\s;\tht)
\bigr)^\top.
\)} of 
the game evaluated at the observed profile $\xhatj$. This 
loss measures how far $\xhatj$ is from satisfying the 
variational inequality induced by $\tht$, which encodes the 
first-order Nash stationarity conditions. It vanishes if 
and only if $\xhatj$ is an equilibrium under $\tht$.

The predictability loss measures instead the squared distance between the observed profile and the nearest 
equilibrium predicted by the model under $\tht$:
\[
\begin{split}
\lpred(\tht;\xhatj,\sj)
:=
&\min_{\x^\star \in \Xj}
\|\xhatj-\x^\star\|^2\\
\quad
&\text{s.t.}\ \F(\tht;\x^\star,\sj)^\top(\x-\x^\star)\ge 0,\ \forall \x\in\Xj.
\end{split}
\]
This loss directly targets predictive accuracy and enjoys 
favorable statistical properties \cite{aswani2018inverse}. However, it is largely 
intractable in practice as its evaluation requires solving a mathematical program with equilibrium constraints.
\medskip

Having recalled both reference losses, we are now ready to 
state the main result of this section, which establishes that 
the suboptimality loss is sandwiched between them.

\begin{theorem}[Loss hierarchy]\label{thm:loss_hierarchy}
For every player $i$, let the utility $\Ui(x_i,x_{-i},\s;\tht)$ 
be concave with $L$-Lipschitz continuous gradient in the own 
action $x_i$. Further let the pseudo-gradient $\F(\tht;\cdot,\s)$ 
be $\mu$-strongly monotone on $\Xj$. Then, for every 
observation $\obsj$,
\[
\left(\mu-\tfrac{L}{2}\right)\lpred(\tht;\xhatj,\sj)
\;\le\;
\lsub(\tht;\xhatj,\sj)
\;\le\;
\lvi(\tht;\xhatj,\sj).
\]
\end{theorem}
\begin{proof}
See Appendix.
\end{proof}

Theorem~\ref{thm:loss_hierarchy} establishes that the 
suboptimality loss is sandwiched between the predictability 
loss and the inverse VI loss. The upper bound 
$\lsub\le\lvi$ shows that our loss is never larger than 
the VI residual, and thus constitutes a tighter measure 
of equilibrium consistency. 
The lower bound relates $\lsub$ to the prediction error $\lpred$, showing that whenever the strong monotonicity of the game dominates the curvature of individual utilities ($\mu-L/2>0$), driving $\lsub$ to zero implies $\lpred$ is driven to zero as well. If strong monotonicity fails, equilibria may be nonunique and this guarantee is unavailable; if $\mu\leq L/2$, the bound is uninformative. As shown next, this result can be further sharpened for affine and quadratic utilities.

\begin{corollary}[Affine utilities]\label{cor:affine_case}
Suppose that for each player $i$, the utility is affine in 
the own action when opponents are fixed, that is,
\[
\Ui(x_i,\xhatj_{-i},\sj;\tht)
=
a_i(\tht;\xhatj_{-i},\sj)^\top x_i
+
c_i(\tht;\xhatj_{-i},\sj),
\]
for some $a_i(\cdot)$ and $c_i(\cdot)$.
Then, for every $\obsj$,
\[
\mu\,\lpred(\tht;\xhatj,\sj)
\;\le\;
\lsub(\tht;\xhatj,\sj)
=
\lvi(\tht;\xhatj,\sj).
\]
\end{corollary}
\begin{proof}
See Appendix.
\end{proof}

When utilities are affine in the own action, the suboptimality
loss and the inverse VI loss coincide exactly, with common calibration constant $\mu$, recovering the standard VI bound \cite[Theorem~1]{bertsimas2015data}.

\begin{corollary}[Quadratic games]\label{cor:quadratic_games}
Consider a game with utilities
\[
U_i(\x,\s;\tht)
=
q_i(\s;\tht)^\top x_i
-
\tfrac{1}{2} x_i^\top Q_{ii} x_i
-
\sum_{r\neq i} x_i^\top Q_{ir} x_r,
\]
for some $q_i(\cdot)$, where $Q_{ii}=Q_{ii}^\top\succeq 0$ for all $i$. Let $H$ 
denote the block matrix with blocks $Q_{ir}$, and define
\[
M:=\frac{H+H^\top}{2}-\frac{1}{2}\operatorname{diag}
(Q_{11},\dots,Q_{PP}).
\]
Then, for every observation $\obsj$,
\[
\lambda_{\min}(M)\,\lpred(\tht;\xhatj,\sj)
\;\le\;
\lsub(\tht;\xhatj,\sj)
\;\le\;
\lvi(\tht;\xhatj,\sj),
\]
where $\lambda_{\min}(M)$ is the calibration constant for $\lsub$ and satisfies $\lambda_{\min}(M)\geq\mu-L/2$, thereby tightening the lower bound in Theorem~\ref{thm:loss_hierarchy}.


\end{corollary}
\begin{proof}
See Appendix.
\end{proof}

This result shows that, for quadratic games, the gap between the suboptimality loss 
and the predictability loss is governed by the smallest 
eigenvalue of $M$, a matrix encoding the coupling structure 
of the game. Interestingly, this gives a provably sharper bound than the scalar 
$\mu-L/2$ of Theorem~\ref{thm:loss_hierarchy}. 
\section{Mirror Descent via Best-Response Oracles}\label{sec:algorithm}
We now turn to the problem of minimizing the empirical 
suboptimality loss over $\tht$. A key structural advantage 
of $\lsub$ is that both its evaluation and its optimization 
decompose into independent player-wise best-response problems. 
Specifically, for each observation $j\in\{1,\dots,N\}$ and 
player $i\in\{1,\dots,P\}$, it suffices to compute the unilateral
best-response
\[
x_i^{j,\star}(\tht)\in\argmax_{x_i\in\Xij} 
\Ui(x_i,\xhatj_{-i},\sj;\tht).
\]
Under Assumption~\ref{ass:convex_game}, each such problem 
is a concave maximization problem in $x_i$ and therefore can be solved efficiently. A key advantage of the proposed loss is that its evaluation and optimization require only unilateral best-response computations, rather than repeated solution of a full equilibrium problem. This makes the method attractive in settings where best responses are easy to compute but equilibrium prediction is expensive. 
The  following proposition shows that these same best-response 
problems also yield a valid subgradient of the empirical 
loss, enabling gradient-based optimization.

\begin{proposition}[Subgradient via best-response 
oracles]\label{prop:subgradient}
Let $x_i^{j,\star}(\tht)$ be a maximizer of the player-wise 
best-response problem above. Then a valid subgradient of 
$\lsubfull$ at $\tht$ is
\[
g(\tht)=\frac{1}{N}\sum_{i,j}
\left[
\nabla_{\tht}\Ui\!\big(x_i^{j,\star}(\tht),\xhatj_{-i},\sj;\tht\big)
-
\nabla_{\tht}\Ui\!\big(\xhatj,\sj;\tht\big)
\right].
\]
In the linear-in-parameter case $\Ui(\x,\s;\tht)=\tht^\top\phii(\x,\s)$, 
this reduces to
\[
g(\tht)=\frac{1}{N}\sum_{i, j}
\left[
\phii(x_i^{j,\star}(\tht),\xhatj_{-i},\sj)
-
\phii(\xhatj,\sj)
\right].
\]
\end{proposition}
\begin{proof}
The claim follows from Danskin's theorem 
\cite[Prop.~B.22(b)]{bertsekas2016nonlinear} applied to each 
player-wise maximization term in 
Definition~\ref{def:sub_loss}, followed by linearity of 
summation.
\end{proof}
The subgradient has a natural interpretation: it measures 
the average feature mismatch between the observed actions 
and the unilateral best responses induced by the current 
parameter estimate $\tht$.

\medskip

Building on this, an effective approach is to minimize the empirical loss using 
mirror descent \cite[Sec. 4.2]{bubeck2015convex}. Starting 
from an initialization $\tht_1\in\ThetaSet$, the update 
at iteration $t$ is
\[
\tht_{t+1}
=
\argmin_{\tht\in\ThetaSet}
\left\{
\eta_t \langle g_t,\tht\rangle + B_{\omega}(\tht,\tht_t)
\right\},
\]
where $g_t=g(\tht_t)$, $\{\eta_t\}_{t\ge 1}$ is a stepsize
sequence, and $B_\omega$ is the Bregman divergence induced
by a given mirror map $\omega$. Since $\lsubfull$ is convex by Proposition~\ref{prop:sub_convex},
standard mirror descent guarantees apply: under the usual
boundedness assumptions and for suitable stepsizes, the averaged
iterate $\bar{\tht}_T:=T^{-1}\sum_{t=1}^T\tht_t$ achieves an
$\mathcal{O}(T^{-1/2})$ convergence rate in objective value
\cite[Thm.~4.2]{bubeck2015convex}.

Compared with predictability-based training, each iteration avoids nested equilibrium solves and instead relies only on player-wise best responses. This substantially improves scalability in many structured games where best responses are inexpensive but computing equilibria is costly. Finally, the proposed formulation naturally extends to online settings where observations arrive sequentially and the parameter can be updated incrementally using the same best-response oracle. 

\begin{algorithm}[t]
\caption{Mirror descent for the suboptimality loss}
\begin{algorithmic}[1]
\Require stepsizes $\{\eta_t\}_{t\ge1}$, mirror map 
$\omega$, init. $\tht_1\in\ThetaSet$
\For{$t=1,2,\dots$}
    \For{each observation $j$ and player $i$}
        \State compute the best response
        \[
        x_i^{j,\star}(\tht_t)\in
        \argmax_{x_i\in\Xij}
        \Ui(x_i,\xhatj_{-i},\sj;\tht_t)
        \]
    \EndFor
    \State form the subgradient
    \[
    g_t=
    \frac{1}{N}\sum_{j=1}^{N}\sum_{i=1}^{P}
    \Big[
    \phii(x_i^{j,\star}(\tht_t),\xhatj_{-i},\sj)
    -
    \phii(\xhatj,\sj)
    \Big]
    \]
    \State update
    \[
    \tht_{t+1}
    =
    \argmin_{\tht\in\ThetaSet}
    \left\{
    \eta_t\langle g_t,\tht\rangle
    +
    B_\omega(\tht,\tht_t)
    \right\}
    \]
    \EndFor
\end{algorithmic}
\end{algorithm}

\section{Case Study: Networked Cournot Competition}\label{sec:cournot}

\begin{figure}[t]
    \centering
    \includegraphics[width=0.95\columnwidth]{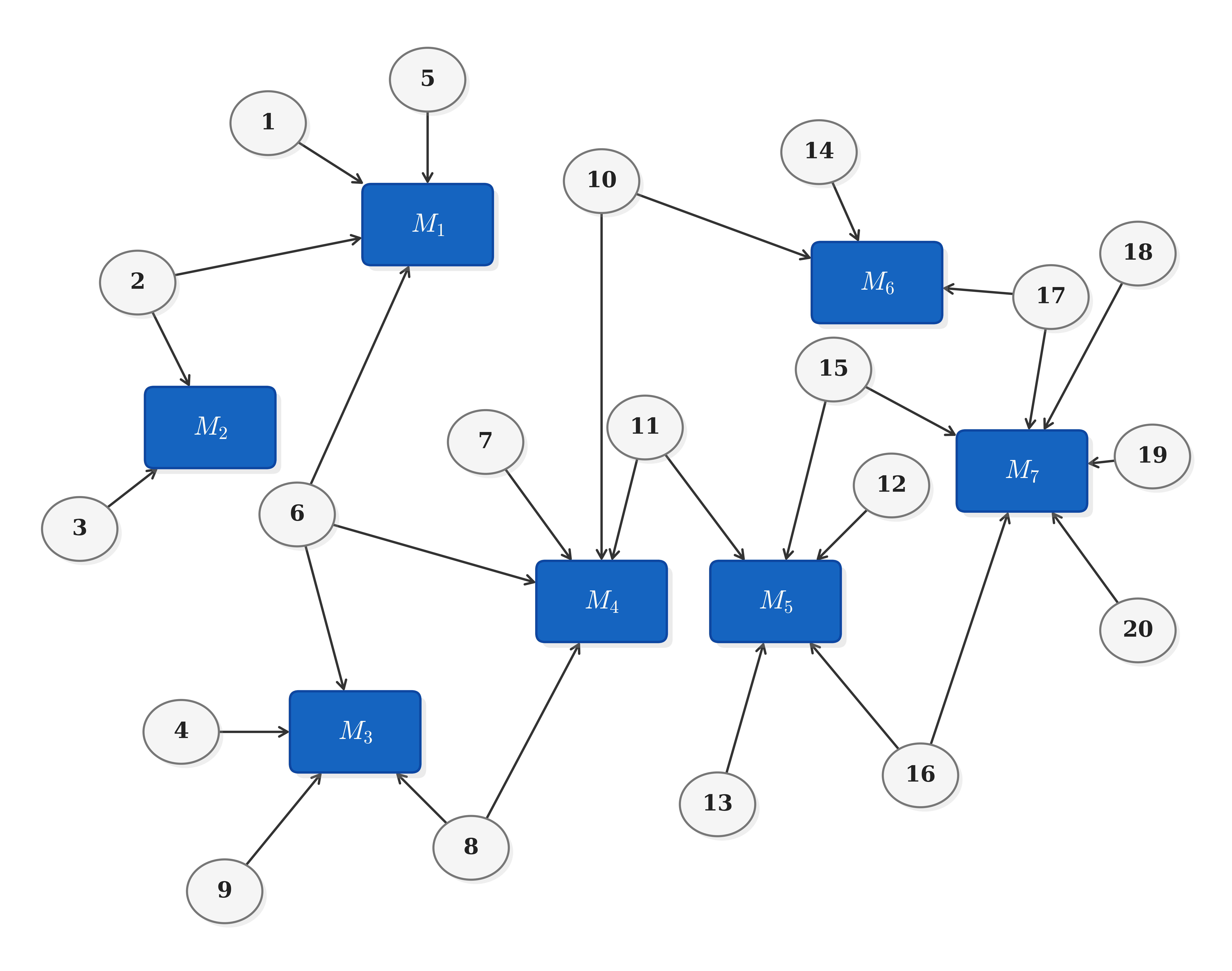}
    \caption{Firm-market participation graph for the networked
    Cournot competition \cite{salehisadaghiani2019distributed}. An edge from firm~$i$ to market~$k$
    indicates that firm~$i$ can supply market~$k$. 
    }
    
    \label{fig:cournot_network}
\end{figure}


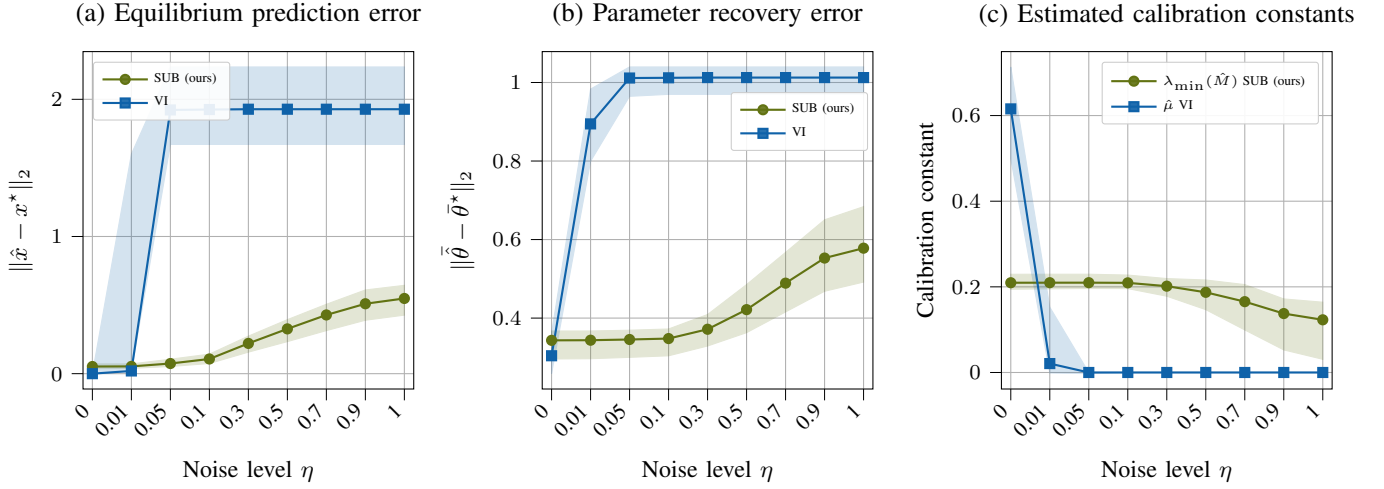
\begin{figure*}[h!]
\centering
\begin{tikzpicture}

\definecolor{darkcyan1598169}{RGB}{15,98,169}
\definecolor{darkgrey176}{RGB}{176,176,176}
\definecolor{darkolivegreen9811519}{RGB}{98,115,19}

\begin{groupplot}[
group style={
    group size=3 by 1,
    horizontal sep=1.7cm
},
width=0.335\textwidth,
height=0.34\textwidth,
tick align=outside
]

\nextgroupplot[
tick pos=left,
title={(a) Equilibrium prediction error},
x grid style={darkgrey176},
xlabel={Noise level $\eta$},
xmajorgrids,
xmin=-0.24, xmax=8.24,
xtick style={color=black},
xtick={0,1,2,3,4,5,6,7,8},
xticklabel style={rotate=45.0,anchor=east},
xticklabels={0,0.01,0.05,0.1,0.3,0.5,0.7,0.9,1},
y grid style={darkgrey176},
ylabel={$\|\hat{x} - x^\star\|_2$},
ymajorgrids,
ymin=-0.111802014005773, ymax=2.34784242835736,
ytick style={color=black},
legend cell align={left},
legend style={
  draw=black!20,
  fill=white,
  fill opacity=0.9,
  text opacity=1,
  rounded corners=1pt,
  font=\tiny,
  at={(0.03,0.97)},
  anchor=north west
}
]

\path [draw=darkolivegreen9811519, fill=darkolivegreen9811519, opacity=0.18]
(axis cs:0,0.0711638733480827)
--(axis cs:0,0.036025445124067)
--(axis cs:1,0.0389421742164561)
--(axis cs:2,0.0522595707646187)
--(axis cs:3,0.0739369277160244)
--(axis cs:4,0.156840654290541)
--(axis cs:5,0.233548489816377)
--(axis cs:6,0.313850744581505)
--(axis cs:7,0.389403599980115)
--(axis cs:8,0.426668155964338)
--(axis cs:8,0.64404317797927)
--(axis cs:8,0.64404317797927)
--(axis cs:7,0.610218441985264)
--(axis cs:6,0.505943605084795)
--(axis cs:5,0.393853370258519)
--(axis cs:4,0.273679137942898)
--(axis cs:3,0.142476871057081)
--(axis cs:2,0.105440747881852)
--(axis cs:1,0.071050342115839)
--(axis cs:0,0.0711638733480827)
--cycle;

\path [draw=darkcyan1598169, fill=darkcyan1598169, opacity=0.18]
(axis cs:0,0.0010601072253219)
--(axis cs:0,6.1016419076764e-09)
--(axis cs:1,0.0152500859829708)
--(axis cs:2,1.6710539992415)
--(axis cs:3,1.67105485246215)
--(axis cs:4,1.67104279019582)
--(axis cs:5,1.67105336046301)
--(axis cs:6,1.67105626454375)
--(axis cs:7,1.67105689604334)
--(axis cs:8,1.6710570279812)
--(axis cs:8,2.23603641164711)
--(axis cs:8,2.23603641164711)
--(axis cs:7,2.23603563679254)
--(axis cs:6,2.23602952109034)
--(axis cs:5,2.23601106851659)
--(axis cs:4,2.23600708695278)
--(axis cs:3,2.23604040824994)
--(axis cs:2,2.23598256805252)
--(axis cs:1,1.60037137024167)
--(axis cs:0,0.0010601072253219)
--cycle;

\addplot [thick, darkolivegreen9811519, mark=*, mark size=1.7, mark options={solid}]
table {%
0 0.052129919744548
1 0.0530082956925592
2 0.0739700544406836
3 0.107078250770341
4 0.220616741070804
5 0.326478303362726
6 0.427993122358961
7 0.509668159286348
8 0.548301823383808
};
\addlegendentry{SUB (ours)}

\addplot [thick, darkcyan1598169, mark=square*, mark size=1.7, mark options={solid}]
table {%
0 1.39414083224614e-05
1 0.0198381234197616
2 1.92256469755143
3 1.92541836103941
4 1.92775932893128
5 1.92776316844459
6 1.92776660690305
7 1.92776661472703
8 1.92776663197712
};
\addlegendentry{VI}

\nextgroupplot[
tick pos=left,
title={(b) Parameter recovery error},
x grid style={darkgrey176},
xlabel={Noise level $\eta$},
xmajorgrids,
xmin=-0.24, xmax=8.24,
xtick style={color=black},
xtick={0,1,2,3,4,5,6,7,8},
xticklabel style={rotate=45.0,anchor=east},
xticklabels={0,0.01,0.05,0.1,0.3,0.5,0.7,0.9,1},
y grid style={darkgrey176},
ylabel={$\|\bar{\hat{\theta}}-\bar{\theta}^{\star}\|_2$},
ymajorgrids,
ymin=0.219391952306932, ymax=1.07869698069472,
ytick style={color=black},
legend cell align={left},
legend style={
  draw=black!20,
  fill=white,
  fill opacity=0.9,
  text opacity=1,
  rounded corners=1pt,
  font=\tiny,
  at={(0.57,0.88)},
  anchor=north west
}
]

\path [draw=darkolivegreen9811519, fill=darkolivegreen9811519, opacity=0.18]
(axis cs:0,0.367104294027932)
--(axis cs:0,0.29584493491368)
--(axis cs:1,0.296510109318819)
--(axis cs:2,0.300008820183215)
--(axis cs:3,0.303928138993597)
--(axis cs:4,0.328859819326786)
--(axis cs:5,0.36284742083525)
--(axis cs:6,0.415655791241899)
--(axis cs:7,0.468274995001838)
--(axis cs:8,0.491728867693464)
--(axis cs:8,0.683816930522776)
--(axis cs:8,0.683816930522776)
--(axis cs:7,0.650320126609635)
--(axis cs:6,0.567295534190888)
--(axis cs:5,0.485319960698078)
--(axis cs:4,0.409158971883246)
--(axis cs:3,0.37227671606615)
--(axis cs:2,0.369182688677946)
--(axis cs:1,0.36719328109761)
--(axis cs:0,0.367104294027932)
--cycle;

\path [draw=darkcyan1598169, fill=darkcyan1598169, opacity=0.18]
(axis cs:0,0.362308998075255)
--(axis cs:0,0.258451271779104)
--(axis cs:1,0.79934626859753)
--(axis cs:2,0.964736381224639)
--(axis cs:3,0.969790405609064)
--(axis cs:4,0.969790401097322)
--(axis cs:5,0.969790403391273)
--(axis cs:6,0.969790405907881)
--(axis cs:7,0.969790406557419)
--(axis cs:8,0.969790406553113)
--(axis cs:8,1.03963766122254)
--(axis cs:8,1.03963766122254)
--(axis cs:7,1.03963766121979)
--(axis cs:6,1.03963766121762)
--(axis cs:5,1.039637661115)
--(axis cs:4,1.03963765904456)
--(axis cs:3,1.03963765885959)
--(axis cs:2,1.03963765430584)
--(axis cs:1,0.983281778962774)
--(axis cs:0,0.362308998075255)
--cycle;

\addplot [thick, darkolivegreen9811519, mark=*, mark size=1.7, mark options={solid}]
table {%
0 0.343333394319477
1 0.34365705200663
2 0.345461194918496
3 0.347859651756811
4 0.371585396816077
5 0.421523477271915
6 0.488939848400838
7 0.552881312170857
8 0.578058604029674
};
\addlegendentry{SUB (ours)}

\addplot [thick, darkcyan1598169, mark=square*, mark size=1.7, mark options={solid}]
table {%
0 0.304085068148413
1 0.894548892296018
2 1.01139431433457
3 1.01194438754802
4 1.01252763865772
5 1.01252764185649
6 1.01252764193102
7 1.01252764189558
8 1.0125276418944
};
\addlegendentry{VI}

\nextgroupplot[
tick pos=left,
title={(c) Estimated calibration constants},
x grid style={darkgrey176},
xlabel={Noise level $\eta$},
xmajorgrids,
xmin=-0.24, xmax=8.24,
xtick style={color=black},
xtick={0,1,2,3,4,5,6,7,8},
xticklabel style={rotate=45.0,anchor=east},
xticklabels={0,0.01,0.05,0.1,0.3,0.5,0.7,0.9,1},
y grid style={darkgrey176},
ylabel={Calibration constant},
ymajorgrids,
ymin=-0.0385689494184654, ymax=0.749364491438027,
ytick style={color=black},
legend cell align={left},
legend style={
  fill opacity=0.92,
  draw opacity=1,
  text opacity=1,
  draw=black!20,
  fill=white,
  rounded corners=1pt,
  font=\tiny,
  at={(.96,0.97)},
  anchor=north east
}
]

\path [draw=darkolivegreen9811519, fill=darkolivegreen9811519, opacity=0.18]
(axis cs:0,0.22938735060085)
--(axis cs:0,0.194309610285341)
--(axis cs:1,0.195364621653533)
--(axis cs:2,0.196110890018302)
--(axis cs:3,0.195881504137672)
--(axis cs:4,0.178106406777585)
--(axis cs:5,0.146726191876798)
--(axis cs:6,0.0993603173085384)
--(axis cs:7,0.0526958805767516)
--(axis cs:8,0.0308049993408165)
--(axis cs:8,0.164179498002222)
--(axis cs:8,0.164179498002222)
--(axis cs:7,0.171796260829534)
--(axis cs:6,0.205294383204872)
--(axis cs:5,0.216030595621997)
--(axis cs:4,0.219383346036584)
--(axis cs:3,0.227580089982268)
--(axis cs:2,0.229506035293347)
--(axis cs:1,0.229471703931131)
--(axis cs:0,0.22938735060085)
--cycle;

\path [draw=darkcyan1598169, fill=darkcyan1598169, opacity=0.18]
(axis cs:0,0.713549335035459)
--(axis cs:0,0.499622687653884)
--(axis cs:1,-0.0012682891635349)
--(axis cs:2,-0.0027537930158976)
--(axis cs:3,-2.36004012016519e-08)
--(axis cs:4,-7.41852626647664e-08)
--(axis cs:5,-1.03868507676513e-07)
--(axis cs:6,-4.2591445488801e-08)
--(axis cs:7,-1.81097162747511e-08)
--(axis cs:8,-1.65585155974828e-08)
--(axis cs:8,0)
--(axis cs:8,0)
--(axis cs:7,0)
--(axis cs:6,0)
--(axis cs:5,-1.57420422721633e-12)
--(axis cs:4,2.52557198068987e-11)
--(axis cs:3,2.07801791512051e-08)
--(axis cs:2,2.8450122492834e-08)
--(axis cs:1,0.152639507554416)
--(axis cs:0,0.713549335035459)
--cycle;

\addplot [thick, darkolivegreen9811519, mark=*, mark size=1.7, mark options={solid}]
table {%
0 0.209593533223566
1 0.209682589270768
2 0.209740881208058
3 0.209327369268631
4 0.201568680808995
5 0.18722896131634
6 0.165409807474173
7 0.137599715617212
8 0.12301608202908
};
\addlegendentry{$\lambda_{\min}(\hat M)$ SUB (ours)}

\addplot [thick, darkcyan1598169, mark=square*, mark size=1.7, mark options={solid}]
table {%
0 0.615591620923055
1 0.0208285528892187
2 2.14040876916697e-09
3 1.05888305641695e-09
4 -3.55359211037046e-09
5 -7.02351799073958e-10
6 -9.95043690518664e-11
7 -1.2077362816549e-11
8 -1.22033982204312e-11
};
\addlegendentry{$\hat\mu$ VI}

\end{groupplot}

\end{tikzpicture}
\caption{Median test prediction error, parameter recovery error (where $\bar \theta := \theta / \| \theta \|$), and estimated calibration constants across the sampled games for varying noise levels in the networked Cournot competition.}
\label{fig:network-cournot-results}
\end{figure*}

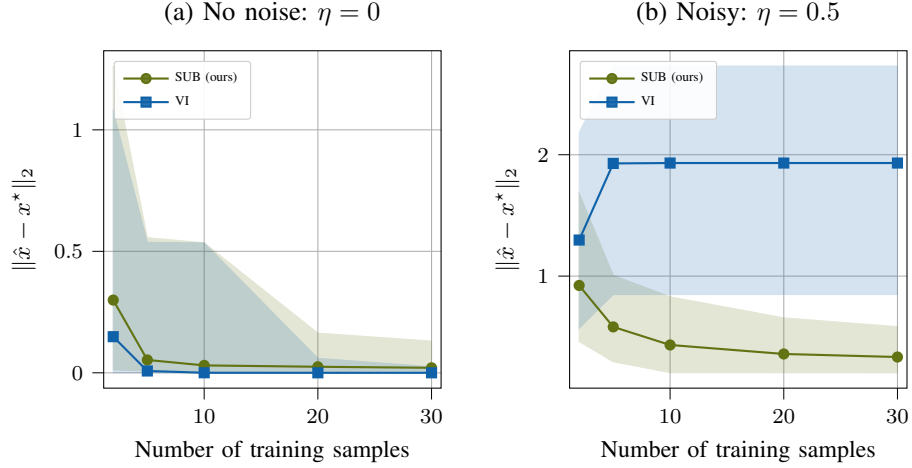
\begin{figure*}[!h]
\centering
\begin{tikzpicture}

\definecolor{darkcyan1598169}{RGB}{15,98,169}
\definecolor{darkgrey176}{RGB}{176,176,176}
\definecolor{darkolivegreen9811519}{RGB}{98,115,19}

\begin{groupplot}[
group style={
    group size=2 by 1,
    horizontal sep=1.7cm
},
width=0.34\textwidth,
height=0.34\textwidth,
tick align=outside
]

\nextgroupplot[
tick pos=left,
title={(a) No noise: $\eta = 0$},
x grid style={darkgrey176},
xlabel={Number of training samples},
xmajorgrids,
xmin=1.16, xmax=30.84,
xtick style={color=black},
y grid style={darkgrey176},
ylabel={$\|\hat{x} - x^\star\|_2$},
ymajorgrids,
ymin=-0.0631245339475627, ymax=1.32561523831633,
ytick style={color=black},
legend cell align={left},
legend style={
  draw=black!20,
  fill=white,
  fill opacity=0.9,
  text opacity=1,
  rounded corners=1pt,
  font=\tiny,
  at={(0.03,0.97)},
  anchor=north west
}
]

\path [draw=darkolivegreen9811519, fill=darkolivegreen9811519, opacity=0.18]
(axis cs:2,1.26249070321342)
--(axis cs:2,0.0119167458968455)
--(axis cs:5,0.0075395955256614)
--(axis cs:10,0.0056694305433315)
--(axis cs:20,0.0052312484708535)
--(axis cs:30,0.0033881761407952)
--(axis cs:30,0.130724266014147)
--(axis cs:30,0.130724266014147)
--(axis cs:20,0.163282037831006)
--(axis cs:10,0.535293865605557)
--(axis cs:5,0.556322528606608)
--(axis cs:2,1.26249070321342)
--cycle;

\path [draw=darkcyan1598169, fill=darkcyan1598169, opacity=0.18]
(axis cs:2,1.07817923646041)
--(axis cs:2,3.84461516835866e-08)
--(axis cs:5,2.3379619158453e-09)
--(axis cs:10,1.35195222154247e-09)
--(axis cs:20,1.43020915913016e-09)
--(axis cs:30,1.15534133186463e-09)
--(axis cs:30,0.024747422062564)
--(axis cs:30,0.024747422062564)
--(axis cs:20,0.0597115326565692)
--(axis cs:10,0.534638725916998)
--(axis cs:5,0.537349199922373)
--(axis cs:2,1.07817923646041)
--cycle;

\addplot [thick, darkolivegreen9811519, mark=*, mark size=1.7, mark options={solid}]
table {%
2 0.299328984053857
5 0.0529433411625737
10 0.0304980054126748
20 0.0248721859291167
30 0.0206454185014214
};
\addlegendentry{SUB (ours)}

\addplot [thick, darkcyan1598169, mark=square*, mark size=1.7, mark options={solid}]
table {%
2 0.14868872742866
5 0.00755088733573445
10 0.00020177443760725
20 1.79943886483282e-05
30 2.52226211082164e-06
};
\addlegendentry{VI}

\nextgroupplot[
tick pos=left,
title={(b) Noisy: $\eta = 0.5$},
x grid style={darkgrey176},
xlabel={Number of training samples},
xmajorgrids,
xmin=1.16, xmax=30.84,
xtick style={color=black},
y grid style={darkgrey176},
ylabel={$\|\hat{x} - x^\star\|_2$},
ymajorgrids,
ymin=0.0769110839131103, ymax=2.85683010091739,
ytick style={color=black},
legend cell align={left},
legend style={
  draw=black!20,
  fill=white,
  fill opacity=0.9,
  text opacity=1,
  rounded corners=1pt,
  font=\tiny,
  at={(0.03,0.97)},
  anchor=north west
}
]

\path [draw=darkolivegreen9811519, fill=darkolivegreen9811519, opacity=0.18]
(axis cs:2,1.6820584283843)
--(axis cs:2,0.460815422435691)
--(axis cs:5,0.294425734517544)
--(axis cs:10,0.204531454155378)
--(axis cs:20,0.203937991232572)
--(axis cs:30,0.203271039231487)
--(axis cs:30,0.583382852677698)
--(axis cs:30,0.583382852677698)
--(axis cs:20,0.656127757370823)
--(axis cs:10,0.829455562621282)
--(axis cs:5,1.00639149239649)
--(axis cs:2,1.6820584283843)
--cycle;

\path [draw=darkcyan1598169, fill=darkcyan1598169, opacity=0.18]
(axis cs:2,2.18324797919418)
--(axis cs:2,0.57048244083674)
--(axis cs:5,0.848013714035051)
--(axis cs:10,0.848012701661931)
--(axis cs:20,0.848013441414038)
--(axis cs:30,0.848013491523841)
--(axis cs:30,2.73047014559901)
--(axis cs:30,2.73047014559901)
--(axis cs:20,2.73046034327259)
--(axis cs:10,2.7304541611605)
--(axis cs:5,2.73046639989925)
--(axis cs:2,2.18324797919418)
--cycle;

\addplot [thick, darkolivegreen9811519, mark=*, mark size=1.7, mark options={solid}]
table {%
2 0.922978864095225
5 0.581530353802681
10 0.432687800394132
20 0.358518930139889
30 0.333379597537162
};
\addlegendentry{SUB (ours)}

\addplot [thick, darkcyan1598169, mark=square*, mark size=1.7, mark options={solid}]
table {%
2 1.29725697811698
5 1.92849003264366
10 1.93164861203122
20 1.93166944988867
30 1.9316703777819
};
\addlegendentry{VI}

\end{groupplot}

\end{tikzpicture}
\caption{Median test prediction error across different sizes of training datasets for VI and SUB.}
\label{fig:network-cournot-sample-sizes}
\end{figure*}

\subsection{Experimental Setup}

To demonstrate the proposed suboptimality-based estimator, we consider the 
well-known networked Cournot competition model on a bipartite firm-market graph, see 
\cite{abolhassani2014network,lin2017networked,bimpikis2019cournot,salehisadaghiani2019distributed}. Specifically, firms are connected to the markets in which they are allowed to sell, and each firm chooses a production quantity on its incident
edges. Let $\mathcal{M}=\{1,\dots,m\}$ denote the set of markets and
$\mathcal{F}=\{1,\dots,P\}$ the set of firms. The bipartite graph
specifies, for each firm~$i$, the subset of accessible markets
$\mathcal{M}_i \subseteq \mathcal{M}$. Firm~$i$ chooses production
quantities $x_i = (x_{ik})_{k\in\mathcal{M}_i}$ with
$0 \le x_{ik} \le \bar{q}_{ik}$, where $\bar{q}_{ik}$ is a known, exogenous capacity upper bound for link $(i,k)$. The aggregate supply in market~$k$ is $Q_k(x)=\sum_{i:\,k\in\mathcal{M}_i} x_{ik}$.

The context-dependent inverse demand in market~$k$ is
\[
p_{ik}(Q_k, s) = \underbrace{\alpha_{0,k} + \delta_k s_k}_{\alpha_k(s)}
\;-\; \beta_{i,k}\, Q_k,
\]
where $\alpha_{0,k}$ is the baseline demand intercept, $\delta_k$
controls the sensitivity to the contextual signal~$s_k$, and
$\beta_{i,k} > 0$ is the demand slope, which we allow to be firm-specific. In contrast to the homogeneous specification in \cite{abolhassani2014network}, this heterogeneity makes the game non-potential. Each firm aims to maximise profit
\[
U_i(x_i, x_{-i}, s;\,\theta)
= \sum_{k\in\mathcal{M}_i} \left[
  x_{ik}\,p_{ik}\!\big(Q_k,\, s\big)
  - c_{ik}\,x_{ik}
  - \tfrac{\gamma_i}{2}\,x_{ik}^2,\right]
\]
where $c_{ik} \ge 0$ is a link-specific linear production cost
(with $c_{ik}=0$ for inactive links) and $\gamma_i > 0$ is a
firm-level quadratic cost coefficient. Given observations of multiple contexts, corresponding equilibria and the firms' individual output capacity constraints, the goal is to find a set of parameters that can best predict the equilibrium for new, unseen contexts. The full parameter vector to be estimated is therefore
\[
\theta = \big(
  \alpha_0 \in \mathbb{R}^m,
  \delta \in \mathbb{R}^m,
  \beta \in \mathbb{R}^{P \times m}_{\ge 0},
  c \in \mathbb{R}^{P \times m}_{\ge 0},
  \gamma \in \mathbb{R}^P_{>0}
\,\big).
\]

We set $P = 20$ firms and $|\mathcal{M}| = 7$ markets with the network graph shown in Figure~\ref{fig:cournot_network}. Let $
\mathcal{E} := \{(i,k) : k \in \mathcal{M}_i\}$ denote the set of active firm market links. We independently generate 50 ground-truth games. For each game, we draw the capacity bounds
$\bar q_{ik} \sim \mathcal{U}(0.3,\,1.0)$ and keep them fixed across all contextual observations of that game. The unknown utility parameters are drawn independently according to the following:
\[
\begin{aligned}
\alpha_{0,k} &\sim \mathcal{U}(0.8,\,1.2),
& \delta_k &\sim \mathcal{U}(0.1,\,0.3),
& \gamma_i &\sim \mathcal{U}(0.5,\,1.0), \\
\beta_{ik} &\sim \mathcal{U}(0.2,\,0.4),
& c_{ik} &\sim \mathcal{U}(0.1,\,0.4),
& & \forall (i,k)\in\mathcal{E}.
\end{aligned}
\]
For inactive links, we set $c_{ik}=0$ and do not define decision variables outside $\mathcal{E}$. For each sampled game, we draw $N_{\mathrm{train}}$ training contexts and an independent test set of $N_{\mathrm{test}}=50$ contexts according to $s^j \sim_{\text{i.i.d.}} \mathcal{U}([-3,3]^m)$. The corresponding equilibria are computed using the classical projected gradient descent algorithm for convex games. We then perturb the training observations by first adding Gaussian noise according to $\tilde x^j = x^{\star,j} + \eta \frac{\|x^{\star,j}\|_2}{\sqrt{d}} \varepsilon^j$, where $d:=|\mathcal E|$ and $\varepsilon^j\sim\mathcal N(0,I_d)$, and then projecting back onto the feasible set, i.e., $\hat x^j = \Pi_{\mathcal X^j}(\tilde x^j) \in \mathcal X^j$. Since $\mathbb E[\|\varepsilon^j\|_2^2]=d$ and projection onto a closed convex set is non-expansive, it follows that $\mathbb E[\|\hat x^j - x^{\star,j}\|_2^2] \le \eta^2 \|x^{\star,j}\|_2^2$, so $\eta$ controls the noise magnitude relative to the norm of the equilibrium observation $\|x^{\star,j}\|_2$.

In our experiments we compare the VI estimator and the proposed suboptimality-based estimator in terms of out-of-sample equilibrium prediction error and parameter recovery error. Figure~\ref{fig:network-cournot-results} fixes $N_{\mathrm{train}}=20$ and varies the noise level. Panels (a) and (b) report the median prediction and parameter recovery errors across the $50$ sampled games. Panel (c) reports the estimated constants in the respective prediction-error bounds:
$\lambda_{\min}(\hat M)$ for SUB, as established in
Corollary~\ref{cor:quadratic_games}, and $\hat\mu$ for VI, using the standard strong-monotonicity bound $\hat\mu\,\ell_{\mathrm{pred}}\leq\ell_{\mathrm{VI}}$
\cite[Theorem~1]{bertsimas2015data}. Figure~\ref{fig:network-cournot-sample-sizes} instead varies $N_{\mathrm{train}}$ in the low-noise and high-noise regimes. All algorithms and experiments were implemented in Python and are available online.\footnote{Code: \url{https://github.com/afeik/inverse_games}}

\subsection{Discussion}
A few general observations are in order. First, in the near-zero-noise regime, both VI and SUB recover parameter estimates that induce near-perfect equilibrium behavior on the test instances. This suggests that, when the observations are sufficiently accurate, both estimators can identify models that reproduce the correct strategic equilibria. The difference between the two methods emerges as the noise level increases. As shown in Figure~\ref{fig:network-cournot-results}(a-b), VI performs competitively for very small perturbations, but its prediction error rises sharply beyond a critical threshold around $\eta = 5\times 10^{-2}$, whereas SUB degrades much more gracefully and remains accurate over a substantially wider range of noise levels. Panel~(c) is consistent with this behavior. Because VI fits first-order conditions evaluated at noisy profiles, it can reduce residuals from inconsistent observations by flattening the fitted pseudo-gradient, causing $\hat{\mu}$ to collapse near the same transition point; SUB instead evaluates best-response gains, which retain curvature information beyond the noisy local gradient, and its estimated $\lambda_{\min}(\hat M)$ remains positive across the tested noise levels. Consequently, the bound in Corollary~\ref{cor:quadratic_games} remains informative for SUB, whereas the standard VI bound becomes meaningless; for affine own-action utilities, however, $\lsub=\lvi$ by Corollary~\ref{cor:affine_case}, so no such difference is expected.
This is consistent with the observed out-of-sample prediction performance, and is in line with the broader prediction-oriented statistical perspective outlined in \cite{aswani2018inverse}. Finally, the above described behavior can also be seen from Figure \ref{fig:network-cournot-sample-sizes}. When no noise is present, see panel (a), the prediction error for VI and SUB is essentially identical as the dataset size increases. This is in stark contrast to the noisy case, see panel (b), where adding samples has the counter-intuitive effect of \emph{worsening} the performance of VI, while improving that of SUB as expected.

\section{Concluding Remarks}


In this work we introduced a game-theoretic suboptimality loss for inverse learning from equilibrium data. The proposed loss is behaviorally interpretable, convex under linear parametrizations, and optimizable using only player-wise best-response oracles. We established bounds linking it to both the inverse variational inequality loss and the equilibrium prediction error, with sharper characterizations for affine and quadratic games. Finally, we provided numerical evidence that the proposed approach substantially outperforms VI-based fitting under noisy observations, an advantage that persists across different training set sizes. Future work includes extensions to generalized Nash equilibrium settings and larger-scale empirical studies in traffic systems and electricity markets.

\newpage

\bibliographystyle{IEEEtran}
\bibliography{references}
\clearpage
\appendix
\section{Proof of Theorem~\ref{thm:loss_hierarchy}}
\begin{proof}[Proof of Theorem~\ref{thm:loss_hierarchy}]
Fix $\obsj$ and let $\xstar\in\Xj$ satisfy
$\F(\tht;\xstar,\sj)^\top(\x-\xstar)\geq 0 \, \forall \, \x\in\Xj$. For each player $i$, define
$g_i(x_i):=\Ui(x_i,\xhatj_{-i},\sj;\tht)$. Since $g_i$ is concave and $L$-smooth on $\Xij$, for all
$x_i,y_i\in\Xij$,
\begin{align}
g_i(y_i)-g_i(x_i)
&\le \nabla g_i(x_i)^\top(y_i-x_i), \label{eq:app-upper}\\
g_i(y_i)-g_i(x_i)
&\ge \nabla g_i(x_i)^\top(y_i-x_i)
-\tfrac{L}{2}\|y_i-x_i\|^2. \label{eq:app-lower}
\end{align}

\noindent\emph{Upper bound.}
Applying \eqref{eq:app-upper} with $x_i=\hat x_i^j$ and $y_i=x_i$ gives
\[
\Ui(x_i,\xhatj_{-i},\sj;\tht)-\Ui(\xhatj,\sj;\tht)
\le -[\F(\tht;\xhatj,\sj)]_i^\top(x_i-\hat x_i^j).
\]
Hence summing over the players gives
\begin{align*}
\lsub(\tht;\xhatj,\sj)
& \le \max_{\x\in\Xj}\F(\tht;\xhatj,\sj)^\top(\xhatj-\x)
= \lvi(\tht;\xhatj,\sj),
\end{align*}
where we used $\Xj=\X_1^j\times\cdots\times\X_P^j$.\\

\noindent\emph{Lower bound.}
Choosing $x_i=x_i^\star$ in the definition of $\lsub$ and applying \eqref{eq:app-lower} with
$x_i=\hat x_i^j$, $y_i=x_i^\star$ yields
\[
\lsub(\tht;\xhatj,\sj)
\ge \F(\tht;\xhatj,\sj)^\top(\xhatj-\xstar)
-\tfrac{L}{2}\|\xhatj-\xstar\|^2.
\]
Furthermore,
\begin{align*}
&\F(\tht;\xhatj,\sj)^\top(\xhatj-\xstar)=\\
&\big(\F(\tht;\xhatj,\sj)-\F(\tht;\xstar,\sj)\big)^\top(\xhatj-\xstar)\\
&+\F(\tht;\xstar,\sj)^\top(\xhatj-\xstar).
\end{align*}
By $\mu$-strong monotonicity,
\[
\big(\F(\tht;\xhatj,\sj)-\F(\tht;\xstar,\sj)\big)^\top(\xhatj-\xstar)
\ge \mu\|\xhatj-\xstar\|^2,
\]
Since $\xhatj\in\Xj$, the defining inequality gives
$\F(\tht;\xstar,\sj)^\top(\xhatj-\xstar)\geq 0$.
Therefore, $\F(\tht;\xhatj,\sj)^\top(\xhatj-\xstar)
\ge \mu\|\xhatj-\xstar\|^2,$ and thus combining the two bounds gives the stated inequality:
\[
\lsub(\tht;\xhatj,\sj)
\ge \Bigl(\mu-\tfrac{L}{2}\Bigr)\lpred(\tht;\xhatj,\sj).
\]

\end{proof}
\begin{proof}[Proof of Corollary~\ref{cor:affine_case}]
For each player $i$, suppose
\[
\Ui(x_i,\xhatj_{-i},\sj;\tht)
=
a_i(\tht;\xhatj_{-i},\sj)^\top x_i
+
c_i(\tht;\xhatj_{-i},\sj).
\]
Then $\nabla_{x_i}\Ui(x_i,\xhatj_{-i},\sj;\tht)=
a_i(\tht;\xhatj_{-i},\sj),$ so the $i$th block of the pseudo-gradient satisfies
\[
[\F(\tht;\xhatj,\sj)]_i
=
-a_i(\tht;\xhatj_{-i},\sj).
\]
Therefore,
\begin{align*}
\lvi(\tht;\xhatj,\sj)
&=
\max_{\x\in\Xj}
\F(\tht;\xhatj,\sj)^\top(\xhatj-\x) \\
&=
\sum_{i=1}^P
\max_{x_i\in\Xij}
-[\F(\tht;\xhatj,\sj)]_i^\top(\hat x_i^j-x_i) \\
&=
\sum_{i=1}^P
\max_{x_i\in\Xij}
a_i(\tht;\xhatj_{-i},\sj)^\top(x_i-\hat x_i^j).
\end{align*}
Using the affine form of $\Ui$, we obtain
\[
a_i(\tht;\xhatj_{-i},\sj)^\top(x_i-\hat x_i^j)
=
\Ui(x_i,\xhatj_{-i},\sj;\tht)
-
\Ui(\xhatj,\sj;\tht),
\]
because the constant term $c_i(\tht;\xhatj_{-i},\sj)$ cancels. Hence,
\begin{align*}
\lvi(\tht;\xhatj,\sj)
&=
\sum_{i=1}^P
\max_{x_i\in\Xij}
\Big[
\Ui(x_i,\xhatj_{-i},\sj;\tht)
-
\Ui(\xhatj,\sj;\tht)
\Big] \\
&=
\lsub(\tht;\xhatj,\sj).
\end{align*}
\end{proof}

\begin{proof}[Proof of Corollary~\ref{cor:quadratic_games}]
Let $\xstar$ attain the minimum in the definition of
$\lpred(\tht;\xhatj,\sj)$, and define
\[
d:=\xhatj-\xstar,
\qquad
D:=\operatorname{blkdiag}(Q_{11},\dots,Q_{PP}),
\]
and note that, by definition of the suboptimality loss,
\[
\lsub(\tht;\xhatj,\sj)
\ge
\sum_{i=1}^P
\Bigl[
\Ui(x_i^\star,\xhatj_{-i},\sj;\tht)-\Ui(\xhatj,\sj;\tht)
\Bigr].
\]
Since \(x_i^\top Q_{ii}x_i = x_i^\top (Q_{ii}+Q_{ii}^\top)x_i/2\), wlog.\ assume \(Q_{ii}=Q_{ii}^\top\succeq 0\). We use an exact second order Taylor expansion:
\[
\Ui(x_i^\star,\xhatj_{-i},\sj;\tht)-\Ui(\xhatj,\sj;\tht)
\]
\[
=\nabla_{x_i}\Ui(\xhatj,\sj;\tht)^\top(x_i^\star-\xhatj_i)
-\frac12 d_i^\top Q_{ii} d_i.
\]
Using \(\F_i=-\nabla_{x_i}\Ui\) and summing over \(i\), this yields
\[
\lsub(\tht;\xhatj,\sj)
\ge
\F(\tht;\xhatj,\sj)^\top d-\frac12 d^\top D d.
\]

Since $\xhatj\in\Xj$, the variational inequality defining
$\xstar$ gives
\[
\F(\tht;\xhatj,\sj)^\top d
\ge
\bigl(\F(\tht;\xhatj,\sj)-\F(\tht;\xstar,\sj)\bigr)^\top d.
\]
For quadratic games, the pseudo-gradient is affine in \(\x\), with Jacobian \(H\), so
\[
\F(\tht;\xhatj,\sj)-\F(\tht;\xstar,\sj)=Hd.
\]
Substituting this gives
\[
\lsub(\tht;\xhatj,\sj)
\ge
d^\top H d-\frac12 d^\top D d
\]
\[
=d^\top \frac{H+H^\top}{2} d-\frac12 d^\top D d
=d^\top M d\ge\lambda_{\min}(M)\|d\|^2.
\]
Hence,
\[
\lsub(\tht;\xhatj,\sj)\ge
\lambda_{\min}(M)\,\lpred(\tht;\xhatj,\sj).
\]

To compare this with Theorem~\ref{thm:loss_hierarchy}, note that
\[
\mu=\lambda_{\min}\!\left(\frac{H+H^\top}{2}\right),
\;
L=\max_{i=1,\dots,P}\lambda_{\max}(Q_{ii})
=
\lambda_{\max}(D).
\]
Therefore, by Weyl's inequality \cite[Thm.~4.3.1]{horn2012matrix} with $i=j=1$,
\[
\lambda_{\min}(M)
=
\lambda_{\min}\!\left(\frac{H+H^\top}{2}-\frac12 D\right)
\]
\[
\ge
\lambda_{\min}\!\left(\frac{H+H^\top}{2}\right)-\frac12\lambda_{\max}(D)
=
\mu-\frac{L}{2},
\]
so the quadratic lower bound refines the scalar bound in Theorem~\ref{thm:loss_hierarchy}.

\end{proof}




\end{document}